\documentclass[sn-mathphys,Numbered]{sn-jnl}

\usepackage{multirow,xcolor,textcomp,manyfoot,booktabs,algpseudocode,listings}%
\usepackage[title]{appendix}%
\usepackage{amsmath,amsthm,amsfonts,amssymb,geometry,subcaption}
\usepackage{graphicx,float,algorithm,algorithmicx,multirow,mathrsfs}
\usepackage{braket}

\theoremstyle{thmstyleone}%
\newtheorem{theorem}{Theorem}%  
\newtheorem{proposition}[theorem]{Proposition}

\theoremstyle{thmstyletwo}%

\theoremstyle{thmstylethree}%

\begin{document}
	
	\title{Classically Realizable Incompatibility}
	
	\author[1]{\fnm{Songyi} \sur{Liu}}\email{liusongyi@buaa.edu.cn}
	
	\author*[1]{\fnm{Yongjun} \sur{Wang}}\email{wangyj@buaa.edu.cn}
	
	\author[1]{\fnm{Baoshan} \sur{Wang}}\email{bwang@buaa.edu.cn}
	
	\author[1]{\fnm{Yunyi} \sur{Jia}}\email{by2309005@buaa.edu.cn}
	
	\author[1]{\fnm{Chang} \sur{He}}\email{hechang@buaa.edu.cn}
	
	\affil*[1]{\orgdiv{School of Mathematical Sciences}, \orgname{Beihang University}, \orgaddress{ \city{Beijing}, \postcode{100191}, \country{China}}}
	
	\abstract{Incompatibility constitutes a fundamental aspect of quantum mechanics. However, not every quantum observable non-classical property arises from incompatibility, nor can all quantum scenarios be fully captured by incompatibility alone. Within the framework of partial Boolean algebra (pBA), we research the structural properties of incompatibility scenarios. We introduce a unified method to realize any incompatibility scenario via a classical game, and the construction is extendable to any scenario embeddable into a Boolean algebra. The exclusivity graph offers a precise characterization of incompatibility scenarios. We prove that every exclusivity graph is the atom graph of an exclusive pBA, which is embedded into a Boolean algebra. These results provide a necessary condition for exclusivity graphs and a sufficient condition for atom graphs. }
	
	
	\maketitle
	
	\section*{Declarations}
	
	\noindent\textbf{Competing interests.}
	The authors have no relevant financial or non-financial interests to disclose.
	
	\vspace{10pt}
	
	\noindent\textbf{Acknowledgments.}
	The work was supported by National Natural Science Foundation of China (Grant No. 12371016, 11871083) and National Key R\&D Program of China (Grant No. 2020YFE0204200).
	
	\section{Introduction}
	
	In quantum mechanics, observables such as position and momentum are incompatible, whereas all classical observables are compatible. Consequently, incompatibility is regarded as a crucial source of nonclassical phenomena. It has been established that incompatibility is necessary for Bell nonlocality \cite{Brunner2014Bell,Aravinda2015The} and for Kochen–Specker contextuality \cite{Budroni2022Kochen,Kurzynski2012Entropic}. Modern approaches to contextuality rely heavily on the compatibility relations among observables \cite{Abramsky2011sheaf,Kurzynski2012Entropic,Fritz2013Entropic,Adan2014Graph,Acin2015A}, and incompatibility plays a key role in various quantum information applications \cite{Colloquium2023Guhne}.
	
	Nevertheless, incompatibility alone is insufficient to guarantee nonlocality or contextuality \cite{Bene2018Measurement}. Specific conditions, such as the presence of a 4‑cycle induced subgraph of the compatibility graph and a contextual quantum correlation, are required \cite{Xu2019Necessary}. Within Spekkens’ framework of contextuality for non‑ideal measurements \cite{Spekkens2005Contextuality}, it has even been shown that incompatibility is not always necessary for contextuality \cite{Contextuality2023Selby,Khrennikov2021Can}. These observations raise a fundamental question: to what extent does incompatibility contribute to nonclassicality? This question can be divided into two aspects: first, which measurement scenarios can be described purely by compatibility relations; and second, whether incompatibility itself can be realized in classical experiments. Here, “nonclassicality” refers to the inability of a system to be described by classical probability theory or hidden‑variable theories, which is also the notion that nonlocality and contextuality were originally introduced to capture \cite{Bell1964On,Kochen1967The}.
	
	As an example of uncertainty relations \cite{Heisenberg1927Uber}, consider Hermitian operators $A$ and $B$ and an arbitrary quantum state $\ket{\psi}$. The Robertson uncertainty relation \cite{Robertson1934An} states:
	\begin{equation}
		\Delta(A)\Delta(B) \geq \frac{1}{2} \left| \bra{\psi} [A,B] \ket{\psi} \right|,
	\end{equation}
	where $\Delta(A)$ is the standard deviation of $A$, and $[A,B] = AB - BA$ is the commutator. Hence, if $A$ and $B$ do not commute, they cannot be simultaneously determined with arbitrary precision.
	
	We refer to ``incompatibility scenarios'' as measurement scenarios arising solely from compatibility relations. The uncertainty relation induces such a scenario, whose logical structure is described below. For simplicity, let observables $A$ and $B$ be dichotomic, each with outcomes $0$ and $1$; i.e., a $(1,2,2)$ scenario. This scenario has four atom events: $A=0$, $A=1$, $B=0$, and $B=1$, denoted as $a$, $\bar{a}$, $b$, and $\bar{b}$, respectively, where $\bar{a}$ denotes the negation of $a$. These atom events are the atoms of the corresponding partial Boolean algebra (pBA) $\mathcal{A}_{(1,2,2)}$ \cite{Kochen1967The,Van2012Noncommutativity,Budroni2023Classical}, illustrated in Figure~\ref{fig-uncertainty}.
	
	\begin{figure}[H]
		\centering
		\includegraphics[width=0.4\linewidth]{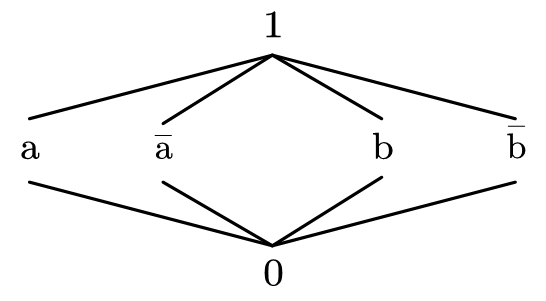}
		\caption{The partial Boolean algebra $\mathcal{A}_{(1,2,2)}$ for two incompatible dichotomic observables $A$ and $B$.}
		\label{fig-uncertainty}
	\end{figure}
	
	A pBA $\mathcal{A}$ is a generalized Boolean algebra equipped with a compatibility relation $\odot \subseteq \mathcal{A} \times \mathcal{A}$, where binary operations are defined only on compatible elements. The algebra of projectors $\mathbb{P}(\mathcal{H})$ on a Hilbert space $\mathcal{H}$, under the commutativity relation, forms a pBA. Thus, pBA captures the logical structure of observable quantum events.
	
	An exclusivity graph \cite{Adan2014Graph} offers an advanced framework for contextuality in any incompatibility scenario. Let $X = \{A_1, \dots, A_n\}$ be the set of observables to measure, and $\mathcal{M}$ be the family of compatible observable sets. A set $C \in \mathcal{M}$ is called a context. A context is maximal if it is not contained in any other context. An atom event $e$ is an assignment $e: C \to O$ (or $e\in O^C$), where $O$ is the outcome set and $C$ is a maximal context. Exclusivity relations arise from compatibility relations: two atom events are exclusive if a joint measurement (within the same scenario) can distinguish them. The exclusivity graph is a simple graph with vertices representing atom events and edges representing exclusivity relations. Vertices serve as the atoms of the scenario. In the $(1,2,2)$ scenario, these are $A=0$, $A=1$, $B=0$, and $B=1$, as shown in Figure~\ref{fig-122}.
	
	\begin{figure}[H]
		\centering
		\includegraphics[width=0.3\linewidth]{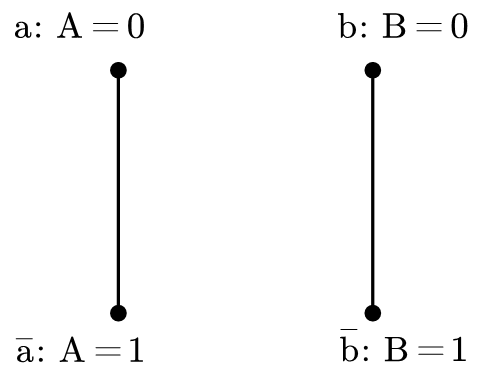}
		\caption{Exclusivity graph of the $(1,2,2)$ scenario.}
		\label{fig-122}
	\end{figure}
	
	Uncertainty relations are counterintuitive from a classical viewpoint, but they alone do not suffice to rule out hidden-variable theories in the historical development of contextuality. This is because incompatibility between observables can be classically realized. Any incompatibility scenario can be modeled by a classical measurement scenario with identical logical structure, as detailed in subsequent sections. Mathematically, any finite incompatibility scenario embeds into a Boolean algebra, implying that non-classical properties of it must depend on contextual states.
	
\section{Classical realization of two-observables incompatibility}\label{sec-122}
	
	Classically, all observables are compatible, but in practice, incompatibility can arise in certain experiments due to incomplete information or operational constraints. Such incompatibility reflects technological limitations, not non-classical properties, yet it provides a method to realize quantum incompatibility with classical scenarios.
	
	To classically realize the $(1,2,2)$ incompatibility scenario, we embed $\mathcal{A}_{(1,2,2)}$ into a Boolean algebra. We list hidden variables that assign values to observables $A$ and $B$. There are four such assignments:
	\begin{equation}
		\begin{split}
			\lambda_1: A,B=0,0&\quad\lambda_2: A,B=0,1 \\
			\lambda_3: A,B=1,0&\quad\lambda_4: A,B=1,1.
		\end{split}
	\end{equation}
	The set $s_d(\mathcal{A}_{(1,2,2)})=\{\lambda_1,\lambda_2,\lambda_3,\lambda_4\}$ is the deterministic states of $\mathcal{A}_{(1,2,2)}$, acting as the hidden-variable space. Since $\mathcal{A}_{(1,2,2)}$ is a finite exclusive pBA, it embeds into a Boolean algebra if and only if it embeds into $\mathcal{A}_{(1,2,2)}^c=\mathcal{P}(s_d(\mathcal{A}_{(1,2,2)}))$, the power-set algebra of its deterministic states \cite{Liu2026The}. Thus, we only need to consider the embedding:
	\begin{equation}
		\begin{split}
			i_{\mathcal{A}_{(1,2,2)}}:\mathcal{A}_{(1,2,2)}&\to \mathcal{A}_{(1,2,2)}^c\\
			0&\mapsto \emptyset \\
			1&\mapsto s_d(\mathcal{A}_{(1,2,2)})\\
			a &\mapsto \{\lambda_1,\lambda_2\}\\
			\bar{a} &\mapsto \{\lambda_3,\lambda_4\}\\
			b &\mapsto \{\lambda_1,\lambda_3\}\\
			\bar{b} &\mapsto \{\lambda_2,\lambda_4\}
		\end{split}
	\end{equation}
	
	It is straightforward to verify that $i_{\mathcal{A}_{(1,2,2)}}$ is an embedding. $\mathcal{A}_{(1,2,2)}^c$ is the minimal Boolean algebra into which $\mathcal{A}_{(1,2,2)}$ embeds \cite{Liu2026The}. This naturally leads to designing a classical experiment (with constraints) to realize the incompatibility scenario $\mathcal{A}_{(1,2,2)}$.

	Any measurement scenario can be formulated as a game \cite{Abramsky2024Combining}. Corresponding to the four deterministic states $\lambda_1,\lambda_2,\lambda_3,\lambda_4$, we design a tetrahedral dice with unmarked faces. Let observables $A$ and $B$ represent the color and material of the faces: red $(A=0)$ or blue $(A=1)$ for color, and gold $(B=0)$ or silver $(B=1)$ for material. The dice is configured as follows:
	\begin{table}[h]
		\centering
		\begin{tabular}{c|c c}
			Face & Color & Material  \\
			\hline
			$\lambda_1$ & red & gold \\
			$\lambda_2$ & red & silver \\
			$\lambda_3$ & blue & gold \\
			$\lambda_4$ & blue & silver \\
		\end{tabular}
	\end{table}
	
	In the game, a referee holds the dice. To measure a property of the downward face, a player first specifies the observable (color or material). The referee then rolls the dice and reports the outcome for that face. The rule enforces that only one observable is measured, and the dice is re-rolled each round, simulating incompatibility. Thus, the incompatibility scenario $\mathcal{A}_{(1,2,2)}$ is classically realized.
	
	Moreover, all probability distributions on $\mathcal{A}_{(1,2,2)}$, forming the state set $s(\mathcal{A}_{(1,2,2)})$, are convex combinations of the deterministic states $\lambda_1,\lambda_2,\lambda_3,\lambda_4$:
	\begin{equation}
		\mathrm{conv}(s_d(\mathcal{A}_{(1,2,2)}))=s(\mathcal{A}_{(1,2,2)}).
	\end{equation}
	This implies that any outcome distributions for $A$ and $B$, including those induced by a quantum state $\ket{\psi}$, can be described classically \cite{Adan2014Graph,Liu2026The}. By adjusting the dice's shape and weight, the game can realize any statistics of the $(1,2,2)$ scenario, showing that uncertainty relations for two observables are classically realizable.
	
	Note that if $\mathcal{A}_{(1,2,2)}$ is treated as an orthomodular lattice in standard quantum logic \cite{Birkhoff1936The}, where binary operations are defined for all elements (even incompatible ones), it cannot embed into Boolean algebras. In standard quantum logic, $a\lor b=1$ and $a\land b=0$. For any homomorphism $f$ from $\mathcal{A}_{(1,2,2)}$ to a Boolean algebra, $f(a)\lor f(b)=1$ and $f(a)\land f(b)=0$ imply $f(b)=\bar{f(a)}=f(\bar{a})$, so $f$ is not an embedding. This likely stems from the lack of direct physical interpretation for operations between incompatible events \cite{Foulis2006Quantum, Kochen2015Reconstruction}. Hence, for physically meaningful events, the pBA framework is more appropriate.
   
   \section{Classically realized incompatibility scenario}\label{sec-general}
   
   This section formalizes the construction from Section~\ref{sec-122} to prove that any finite incompatibility scenario embeds into a Boolean algebra and can thus be realized by a classical game.
   
   Consider the general case. A finite incompatibility scenario consists of:
   \begin{enumerate}
   	\item A set of observables: $X = \{A_1, \dots, A_n\}$,
   	\item A family of contexts: $\mathcal{M} \subseteq \mathcal{P}(X)$,
   	\item An outcome set: $O$,
   \end{enumerate}
   satisfying:
   \begin{itemize}
   	\item If $C \in \mathcal{M}$ and $C' \subseteq C$, then $C' \in \mathcal{M}$.
   \end{itemize}
   
   For simplicity, assume all observables share the same outcome set $O$; this does not affect subsequent conclusions. For any context $C \in \mathcal{M}$, let $O^C$ be its corresponding event set. Define $\mathcal{B}_C$ as the Boolean algebra generated by $O^C$, i.e., $\mathcal{B}_C = \mathcal{P}(O^C)$. If $C' \subseteq C$, then $O^{C'}$ has a marginal representation in $\mathcal{B}_C$:
   \begin{equation}
   	e \in O^{C'} \mapsto \bigvee\{f \in O^C : f|_{C'} = e\}.
   \end{equation}
   Thus, $\mathcal{B}_{C'}$ can be treated as a subalgebra of $\mathcal{B}_C$. We therefore focus on maximal contexts. Let $\mathcal{F}$ be the family of Boolean algebras $\mathcal{B}_C$ and all their Boolean subalgebras, where $C$ is maximal in $\mathcal{M}$. It is straightforward to show that Boolean operations coincide on $\mathcal{B}_{C_1} \cap \mathcal{B}_{C_2}$ for any pair of maximal contexts $C_1$ and $C_2$; indeed, $\mathcal{B}_{C_1} \cap \mathcal{B}_{C_2} = \mathcal{B}_{C_1 \cap C_2}$. Consequently, $\mathcal{F}$ forms a partial Boolean algebra (pBA) with underlying set $\mathcal{B}_{\mathcal{M}} = \bigcup_{C \in \mathcal{M}} \mathcal{B}_C$, where two elements are compatible if and only if they belong to the same Boolean algebra.
   
   In quantum logic and Kochen–Specker contextuality \cite{Van2012Noncommutativity,Abramsky2020The}, pBAs are often assumed to satisfy Specker's principle: pairwise compatibility implies global compatibility \cite{Specker1960Die}. Mathematically, any set of pairwise compatible elements must lie within a Boolean algebra. This definition suits ideal measurements. However, the pBA definition adopted here is the general version \cite{Budroni2023Classical}, which does not require Specker's principle. Thus, our conclusions also hold for more general measurements.
   
   Now we embed the pBA $\mathcal{B}_{\mathcal{M}}$ into a Boolean algebra. The construction is natural. Let $\mathcal{B}_X = \mathcal{P}(O^X)$ be the Boolean algebra of joint measurement outcomes for all observables. For any $C \in \mathcal{M}$, there is a marginal representation of $O^C$ in $\mathcal{B}_X$:
	\begin{equation}
		e \in O^C \mapsto \bigvee\{f \in O^X : f|_{C} = e\}.
	\end{equation}
	This induces an embedding $i_C: \mathcal{B}_C \to \mathcal{B}_X$ for each maximal context $C \in \mathcal{M}$. Define a map $i_{\mathcal{M}}: \mathcal{B}_{\mathcal{M}} \to \mathcal{B}_X$ by:
	\begin{equation}
		\label{eq-embed}
		\begin{split}
			i_{\mathcal{M}}: \mathcal{B}_{\mathcal{M}} &\to \mathcal{B}_X,\\
			e \in \mathcal{B}_C &\mapsto i_C(e).
		\end{split}
	\end{equation}
	For an element $e \in \mathcal{B}_{C_1} \cap \mathcal{B}_{C_2}$, we have $i_{C_1}(e) = i_{C_2}(e) = \bigvee\{f \in O^X : f|_{C_1 \cap C_2} = e\}$, so $i_{\mathcal{M}}$ is well-defined. Since each $i_C$ is a homomorphism, $i_{\mathcal{M}}$ is also a homomorphism. If $e_1 \neq e_2$ are compatible, they belong to some maximal Boolean subalgebra $\mathcal{B}_C$, and $i_{\mathcal{M}}(e_1) \neq i_{\mathcal{M}}(e_2)$ because $i_C$ is an embedding. If $e_1 \neq e_2$ are incompatible, they correspond to distinct observable sets, ensuring $i_{\mathcal{M}}(e_1) \neq i_{\mathcal{M}}(e_2)$. Thus, $i_{\mathcal{M}}$ is an embedding from $\mathcal{B}_{\mathcal{M}}$ into $\mathcal{B}_X$.
	
	In summary:
	
	\begin{proposition}\label{prop-1}
		Any finite incompatibility scenario can be embedded into a classical scenario.
	\end{proposition}
	
	Similar to Section~\ref{sec-122}, a classical game can be designed to realize any finite incompatibility scenario. A referee holds a random system with hidden-variable space $O^X$ (e.g., a dice with $|O^X|$ faces). In each round, the player inputs an observable $A \in X$, and the referee runs the system to output the outcome of $A$. Thus, incompatibility alone does not yield non-classical properties; these arise from contextual states on the scenario.
	
	\section{Quantum scenarios beyond incompatibility}
	
	Proposition~\ref{prop-1} shows that not all quantum scenarios are describable by incompatibility alone. The first class of exceptions are scenarios that cannot embed into Boolean algebras, such as Kochen-Specker scenarios, which admit no global truth-value assignments. Typical examples include Kochen-Specker vector sets \cite{Kochen1967The,Cabello1997Bell} and the Peres-Mermin square \cite{Peres1991Two,Mermin1990Simple}.
	
	\begin{figure}[H]
		\centering
		\includegraphics[width=0.7\linewidth]{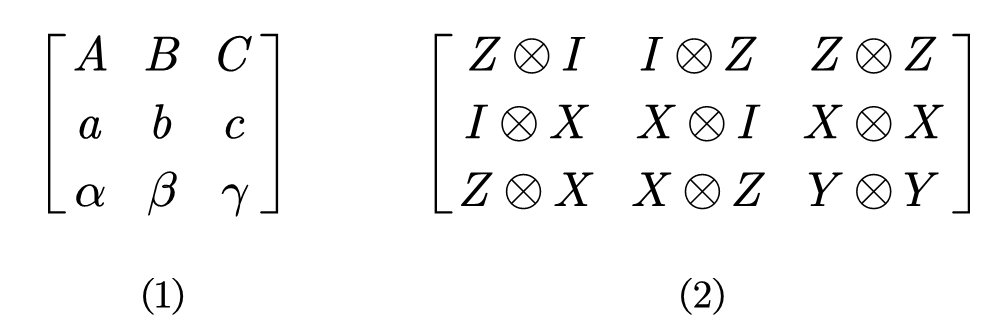}
		\caption{(1) The Peres-Mermin square, where each row or column forms a context; (2) A quantum realization of the Peres-Mermin square, which is a Kochen-Specker scenario, with $X, Y, Z$ as Pauli operators on qubits.}
		\label{fig-PM}
	\end{figure}
	
	The Peres-Mermin square consists of nine observables with outcomes $\pm 1$. Their compatibility relations are shown in Figure~\ref{fig-PM}(1), but this structure alone is insufficient to define a Kochen-Specker scenario. Additional conditions on these observables are required, such as:
	\begin{equation}
		\label{eq-PM}
		\begin{split}
			&ABC = 1,\quad abc = 1,\quad \alpha\beta\gamma = 1,\\
			&Aa\alpha = 1,\quad Bb\beta = 1,\quad Cc\gamma = -1.
		\end{split}
	\end{equation}
	The contextuality of the Peres-Mermin scenario stems from these conditions. They are realized by a quantum scenario in Figure~\ref{fig-PM}(2), which admits no classical realization. If the nine observables could be assigned simultaneous values, the products $\{ABC, abc, \alpha\beta\gamma, Aa\alpha, Bb\beta, Cc\gamma\}$ would contain an even number of $+1$ outcomes, contradicting the conditions.
	
	Conditions like equations~(\ref{eq-PM}) provide structural information about the scenario, including which events are impossible and which are equivalent. This information cannot be captured by an incompatibility scenario alone. In fact, we can present a simple quantum scenario that is classically realizable yet cannot be described by an incompatibility scenario.
	
	Define five vectors:
	\begin{align*}
		\ket{a}&= (1,0,0)^T, \quad \ket{b} = (0,1,0)^T, \quad \ket{c} = (0,0,1)^T, \\
		\ket{d}&= (1/\sqrt{2},1/\sqrt{2},0)^T, \quad \ket{e} = (1/\sqrt{2},-1/\sqrt{2},0)^T.
	\end{align*}
	In quantum mechanics, these five vectors (projectors) represent five quantum events, denoted $a,b,c,d,e$. They generate a quantum scenario $\mathcal{Q}$, whose atom graph is shown in Figure~\ref{fig-5events}.
	\begin{figure}[H]
		\centering
		\includegraphics[width=0.3\linewidth]{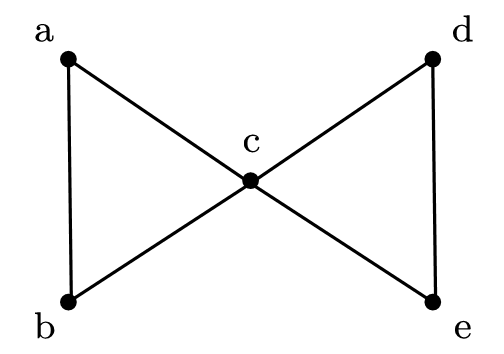}
		\caption{Atom graph of the five-atom scenario $\mathcal{Q}$. Two vertices are adjacent if the corresponding events are compatible and exclusive.}
		\label{fig-5events}
	\end{figure}
	
	Specifically, the scenario consists of five atom events $\{a,b,c,d,e\}$. The context $\{a,b,c\}$ corresponds to the outcomes of a three-valued observable $A$, and $\{c,d,e\}$ to those of another three-valued observable $B$, where $A$ and $B$ are incompatible. If $c$ represents, say, $A=0$ and $B=0$, these are equivalent because they correspond to the same event $c$. In summary, Figure~\ref{fig-5events} depicts the exclusivity and equivalence relations among atom events, which suffice to capture the logical structure of any exclusive partial Boolean algebra \cite{Liu2025Atom}.
	
	Figure~\ref{fig-5events} cannot be described by an incompatibility scenario. It contains only two maximal contexts of three events, which can only arise from two incompatible three-valued observables. The corresponding exclusivity graph would then be isomorphic to that in Figure~\ref{fig-2observables}. Without additional conditions, one cannot determine whether $A=0$ is equivalent to $B=0$.
	\begin{figure}[H]
		\centering
		\includegraphics[width=0.4\linewidth]{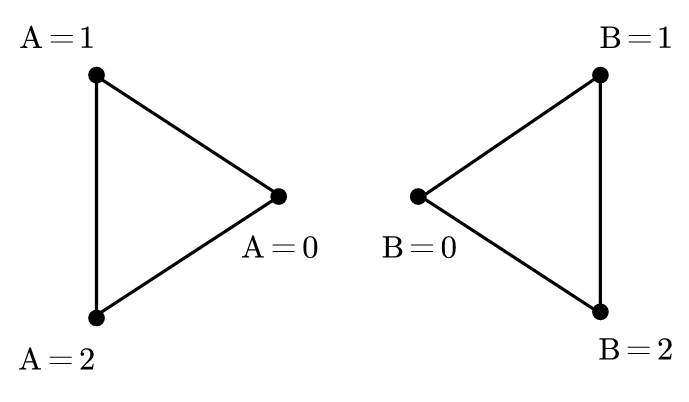}
		\caption{Exclusivity graph for two incompatible three-valued observables $\{A,B\}$.}
		\label{fig-2observables}
	\end{figure}
	
	A classical game can realize the scenario in Figure~\ref{fig-5events}. The construction from Section~\ref{sec-general} applies to any scenario embeddable into a Boolean algebra. Here, since $\mathcal{Q}$ admits five deterministic states, we design a five-sided dice labeled $\lambda_1,\dots,\lambda_5$. Each face has a color (red, yellow, or blue) and a material (gold, silver, or copper), as shown below:
	\begin{table}[h]
		\centering
		\begin{tabular}{c|c c}
			Face & Color & Material  \\
			\hline
			$\lambda_1$ & red & gold \\
			$\lambda_2$ & red & silver \\
			$\lambda_3$ & yellow & gold \\
			$\lambda_4$ & yellow & silver \\
			$\lambda_5$ & blue & copper
		\end{tabular}
	\end{table}
	
	A referee holds this dice. Following the same procedure as in Section~\ref{sec-general}, we obtain a classical realization of scenario $\mathcal{Q}$.
	
	\section{Exclusivity Graph and Logical Exclusivity Principle}
	
	This section examines a precise graphical representation for incompatibility scenarios: the exclusivity graph \cite{Adan2014Graph}. Given a scenario $(X, \mathcal{M}, O)$, two events $e_1 \in O^{C_1}$ and $e_2 \in O^{C_2}$ are exclusive if there exists an observable $A \in C_1 \cap C_2$ such that $e_1(A) \neq e_2(A)$. The graph with atom events as vertices and exclusivity relations as edges is the scenario's exclusivity graph.
	
	Consider a simple example with $\mathcal{M} = \{\{A,B\}, \{B,C\}\}$ and $O = \{0,1\}$. Its compatibility graph and corresponding exclusivity graph are shown in Figure~\ref{fig-EG}. This structure is a key component in the exclusivity graphs of n-cycle scenarios, such as Clauser-Horne-Shimony-Holt (CHSH) \cite{Clauser1969Proposed} and Klyachko-Can-Binicio\u{g}lu-Shumovsky (KCBS) \cite{Klyachko2008Simple}.
	
	\begin{figure}[H]
		\centering
		\includegraphics[width=0.8\linewidth]{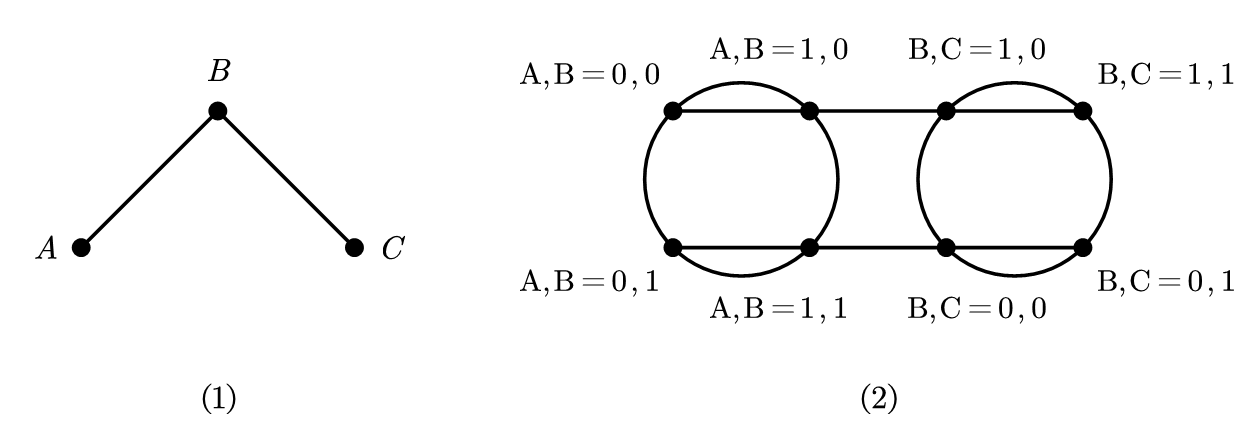}
		\caption{(1) Compatibility graph for $\mathcal{M} = \{\{A,B\}, \{B,C\}\}$; (2) Corresponding exclusivity graph, drawn as a hypergraph where each line or circle denotes a maximal clique.}
		\label{fig-EG}
	\end{figure}
	
	In Figure~\ref{fig-EG}(2), besides the cliques for the maximal contexts $\{A,B\}$ and $\{B,C\}$, exclusivity induces two additional cliques (shown as lines). They are introduced to depict the probability exclusivity principle (PEP) \cite{Fritz2013Local}: the sum of probabilities for pairwise exclusive events cannot exceed $1$. Thus, the exclusivity graph offers a more precise description of an incompatibility scenario.
	
	In quantum mechanics, exclusive events must be compatible, known as the logical exclusivity principle (LEP) \cite{Abramsky2020The}. In the example of Figure~\ref{fig-EG}, although the joint measurements $\{A,B\}$ and $\{B,C\}$ are incompatible, the coarse-grained measurements for events ``$A,B=i,0$" and ``$B,C=1,j$" (with $i,j \in \{0,1\}$) are compatible; equivalently, their corresponding operator product is zero. It is proven that LEP is stronger than PEP and is equivalent to transitivity \cite{Abramsky2020The}.
	
	Formally, for a partial Boolean algebra (pBA) $\mathcal{A}$, elements $a, b \in \mathcal{A}$ are exclusive if there exists $c \in \mathcal{A}$ such that $a \leq c$ and $b \leq \bar{c}$. If every pair of exclusive elements is compatible, then $\mathcal{A}$ is said to satisfy LEP, and is called an exclusive partial Boolean algebra (epBA).
	
	We research whether the exclusivity graph captures the logical exclusivity principle (LEP) and its ability to describe measurement scenarios. Let $\mathcal{G}$ be the exclusivity graph of an incompatibility scenario $(X,\mathcal{M},O)$. Our goal is to characterize the partial Boolean algebra (pBA) generated by $\mathcal{G}$. Constructing a pBA is generally challenging. Fortunately, note that the exclusivity graph introduces no new events. As shown in Section~\ref{sec-general}, all events in $\mathcal{B}_{\mathcal{M}}$ can be embedded into a Boolean algebra $\mathcal{B}_X$. This leads to:
	
	\begin{proposition}
		\label{prop-2}
		Every exclusivity graph $\mathcal{G}$ generates a partial Boolean algebra $\mathcal{B}_{\mathcal{G}}$, and $\mathcal{B}_{\mathcal{G}}$ can be embedded into a Boolean algebra.
	\end{proposition}
	\begin{proof}
		Using embedding \eqref{eq-embed} from Section~\ref{sec-general}, for each maximal clique $C$ of $\mathcal{G}$, let $\mathcal{B}_C$ be the Boolean subalgebra of $\mathcal{B}_X$ generated by $i_{\mathcal{M}}(C)$. Take $\mathcal{F}_{\mathcal{G}}$ as the family of all such $\mathcal{B}_C$ and their Boolean subalgebras. Then $\mathcal{F}_{\mathcal{G}}$ forms a pBA $\mathcal{B}_{\mathcal{G}}$, which is in fact a partial Boolean subalgebra of $\mathcal{B}_X$.
	\end{proof}
	
	Proposition~\ref{prop-2} shows that the pBA $\mathcal{B}_{\mathcal{G}}$ generated by the exclusivity graph captures the logical structure more precisely than the original incompatibility scenario $\mathcal{B}_{\mathcal{M}}$. Clearly, $\mathcal{B}_{\mathcal{M}} \subseteq \mathcal{B}_{\mathcal{G}}$. Nevertheless, even when described by exclusivity graphs, incompatibility scenarios remain embeddable into classical scenarios.
	
	We now determine whether $\mathcal{B}_{\mathcal{G}}$ is an exclusive partial Boolean algebra (epBA). An epBA describes a no-signaling scenario enabling normal logical reasoning. Determining whether a given simple graph is the atom graph of an epBA is a significant and often challenging problem. The following proposition resolves this for exclusivity graphs.
	
	\begin{proposition}
		$\mathcal{B}_{\mathcal{G}}$ satisfies the LEP. Therefore, every exclusivity graph is the atom graph of a finite epBA.
	\end{proposition}
	\begin{proof}
		Since $\mathcal{B}_{\mathcal{G}}$ is a partial Boolean subalgebra of $\mathcal{B}_X$, its elements can be viewed as subsets of $O^X$.
		
		Consider two elements $e = \bigcup_{i \in I} a_i$ and $f = \bigcup_{j \in J} b_j$ in $\mathcal{B}_{\mathcal{G}}$, where $a_i \in i_{\mathcal{M}}(C_1)$, $b_j \in i_{\mathcal{M}}(C_2)$, and $C_1, C_2$ are maximal cliques of $\mathcal{G}$. If $e$ and $f$ are exclusive, they are disjoint. Hence, for any $i \in I$ and $j \in J$, $a_i$ and $b_j$ are exclusive. Furthermore, for distinct $i, i' \in I$ and $j, j' \in J$, $a_i$ is exclusive with $a_{i'}$, and $b_j$ is exclusive with $b_{j'}$. Therefore, all elements in $\{a_i, b_j\}_{i\in I, j\in J}$ are pairwise exclusive, implying their preimages under $i_{\mathcal{M}}$ belong to the same maximal clique of $\mathcal{G}$. Thus, $a_i$ and $b_j$ lie within the same maximal Boolean subalgebra of $\mathcal{B}_{\mathcal{G}}$, proving that $e$ and $f$ are compatible.
	\end{proof}
	
	For scenarios satisfying Specker's principle, such as ideal measurements, a crucial result establishes that any finite exclusive partial Boolean algebra (epBA) is uniquely determined by its atom graph \cite{Liu2025Atom}. Specifically, two finite epBAs $\mathcal{A}_1$ and $\mathcal{A}_2$ are isomorphic if and only if their atom graphs are isomorphic, yielding a one-to-one correspondence. However, for the general scenario, which does not assume Specker's principle, it remains unclear whether this correspondence still holds.
	
	\section{Conclusion and Outlook}\label{sec-conclusion}
	
	This work investigates the extent to which incompatibility contributes to nonclassicality within the framework of partial Boolean algebras (pBAs). We demonstrate that any incompatibility scenario can be classically realized. A unified method is provided to realize such scenarios via a classical game, which can be generalized to any scenario embeddable into a Boolean algebra. Consequently, incompatibility alone does not yield nonclassical properties without contextual states. Incompatibility scenarios cannot describe those scenarios resisting Boolean embedding, such as Kochen-Specker scenarios, nor can they fully capture certain quantum scenarios, even those that are Boolean-embeddable.
	
	The exclusivity graph offers a precise and advanced description of incompatibility scenarios, playing a key role in contextuality research. We prove that every exclusivity graph is the atom graph of a pBA satisfying the logical exclusivity principle (LEP), and that the scenario it generates can be embedded into a Boolean algebra. These results establish a necessary condition for which scenarios can be described by exclusivity graphs, and a sufficient condition for identifying atom graphs among simple graphs.
	
	While we provide a necessary condition for exclusivity graphs, a full characterization, i.e., a necessary and sufficient condition for a scenario to be describable by an exclusivity graph, remains open. Similarly, a convenient necessary and sufficient condition for recognizing atom graphs is still lacking. Finally, the potential nonclassicality of an incompatibility scenario stems from its associated set of contextual states. Characterizing the contextual, and specifically quantum, state sets for a given incompatibility scenario remains a significant challenge.
	
	\bibliography{math}

\end{document}